\documentclass{article}

\usepackage{xspace,xcolor}
\usepackage{graphicx}
\usepackage{mdwlist}
\usepackage{amssymb,amsfonts,amsmath,amsthm}
\usepackage{hyperref}

\usepackage[margin=1.0in]{geometry}
\usepackage[noend,ruled]{algorithm2e}

\title{On Romeo and Juliet Problems: Minimizing Distance-to-Sight\footnote{This work by Ahn and Oh was supported by the MSIT (Ministry of Science and ICT), Korea, under the SW Starlab support program (IITP-2017-0-00905) supervised by the IITP (Institute for Information \& Communications Technology Promotion).
\newline An extended abstract of this paper appeared at SWAT'18.}}

\newtheorem*{problem*}{Problem}
\newtheorem{lemma}{Lemma}
\newtheorem{theorem}{Theorem}
\newtheorem{corollary}{Corollary}
\newtheorem{definition}{Definition}
\newcommand{\ccheck}[1]{#1}
\newcommand{\boundaryOfP}{\ensuremath{\partial P}\xspace}

\newcommand{\fig}[0]{Figure}
\begin{document}

\newcommand{\propone}{Property~\textbf{P1}}
\newcommand{\proptwo}{Property~\textbf{P2}}
\newcommand{\propthree}{Property~\textbf{P3}}

\newcommand{\email}[1]{\href{mailto:#1}{#1}}

  \author{Hee-Kap Ahn\thanks{Department of Computer Science and Engineering, POSTECH, Pohang, South Korea, \email{heekap@postech.ac.kr}}
  \and
  Eunjin Oh\thanks{Max Planck Institute for Informatics, Saarbr\"ucken, Germany, \email{eoh@mpi-inf.mpg.de}}
  \and
  Lena Schlipf\thanks{Theoretische Informatik, FernUniversit\"at in Hagen, Hagen, Germany, \email{lena.schlipf@fernuni-hagen.de}}
  \and
  Fabian Stehn\thanks{Institut f\"ur Informatik, Universit\"at Bayreuth, Bayreuth, Germany, \email{fabian.stehn@uni-bayreuth.de}}
  \and
  Darren Strash\thanks{Department of Computer Science, Hamilton College, Clinton, New York, USA, \email{dstrash@hamilton.edu}}
  }

\date{}
\maketitle

  \begin{abstract} We introduce a variant of the watchman route
    problem, which we call the \emph{quickest pair-visibility}
    problem.  Given two persons standing at points $s$ and $t$ in a
    simple polygon $P$ with no holes, we want to minimize the distance
    they travel in order to see each other in $P$. We solve two
    variants of this problem, one minimizing the longer distance the
    two persons travel (min-max) and one minimizing the total travel
    distance (min-sum), optimally in linear time.  We also consider a
    query version of this problem for the min-max variant. We can
    preprocess a simple $n$-gon in linear time so that the minimum of
    the longer distance the two persons travel can be computed in
    $O(\log^2 n)$ time for any two query positions $s,t$ where the two
    persons start.\\

\noindent\textbf{Keywords:} visibility polygon $\cdot$ shortest path $\cdot$ watchman problems
  \end{abstract}

\section{Introduction}
In the watchman route problem, a watchman takes a route to
\emph{guard} a given region---that is, any point in the region is
visible from at least one point on the route. It is desirable to make
the route as short as possible so that the entire area can be guarded
as quickly as possible. The problem was first introduced in 1986 by
Chin and Ntafos~\cite{Chin1986} and has been extensively studied in
computational geometry~\cite{Carlsson1999,Mitchell2013}.
Though the problem is NP-hard for polygons with
holes~\cite{Chin1986,Chin1988,Dumitrescu2012}, an optimal route can be
computed in time $O(n^3\log n)$ for simple $n$-gons~\cite{Dror2003}
when the tour must pass through a specified point, and $O(n^4\log n)$
time otherwise.
	
In this paper, we study a variant of the watchman route
  problem. Imagine two persons, Romeo and Juliet, travel in a region
  from their starting locations.
  They want to minimize the distance they travel in order to see each
  other. More precisely, given the region and the locations where
  Romeo and Juliet start, the objective is to compute their paths, one
  for Romeo and one for Juliet, such that they see each other after
  traveling along the paths and their travel distances are minimized.
  This problem can be formally defined as follows.
\begin{problem*}[quickest pair-visibility problem] Given two points
  $s$ and $t$ in a simple polygon $P$, compute the minimum distance
  that $s$ and $t$ must travel in order to see each other in $P$.
\end{problem*}
In the \textit{min-max} variant of the quickest
  pair-visibility problem, we want to minimize the longer distance
  that the two points travel to see each other.  In the
  \textit{min-sum} variant, we want to minimize the total travel
  distance that the two points travel to see each other.
	
This problem may sound similar to the shortest path problem between
$s$ and $t$, in which the objective is to compute the shortest path
$\pi(s,t)$ for $s$ to \emph{reach} $t$. However, they differ even for
a simple case: for any two points lying in a convex polygon, the
distance in the quickest pair-visibility problem is zero while in the
shortest path problem, it is their geodesic distance
$|\pi(s,t)|$. We would like to mention that our algorithm
  to be presented later uses the shortest path as a guide in computing
  the quickest pair-visibility paths.
		
The quickest pair-visibility problem occurs in optimization tasks.
For example, mobile robots that use a line-of-sight communication
model are required to move to mutually-visible positions to establish
communication~\cite{Ganguli2007}.  An optimization task here is to
find shortest paths for the robots to meet the visibility requirement
for establishing communication among them.
	
Wynters and Mitchell~\cite{Wynters1993} studied this problem for two
agents acting in a polygonal domain in the presence of polygonal
obstacles and gave an $O(nm)$-time algorithm for the min-sum variant
(where $n$ is the number of vertices of the polygonal domain, and $m$
is the number of edges of the visibility graph of all corners) and an
$O(n^3 \log{n})$-time algorithm for the min-max variant.
	
A query version of the quickest visibility problem has also been
studied~\cite{Arkin2016,Khosravi2005,Wang2017}.  In the query problem,
a polygon and a source point lying in the polygon are given, and the
goal is to preprocess them and construct a data structure that supports,
for a given query point, finding the shortest path taken from the
source point to see the query point efficiently. Khosravi and
Ghodsi~\cite{Khosravi2005} considered the case for a simple $n$-gon
and presented an algorithm to construct a data structure of $O(n^2)$
space so that, given a query, it finds the shortest visibility path in
$O(\log n)$ time.  Later, Arkin et al.~\cite{Arkin2016} improved the
result and presented an algorithm for the problem in a polygonal
domain.  Very recently, Wang~\cite{Wang2017} presented an improved
algorithm for this problem for the case that the number of the holes
in the polygon is relatively small.  Figure~\ref{fig:problems}(a)
illustrates differences in these problems for a simple polygon and two
points, $s$ and $t$, in the polygon.

\subsection{Our results}
In this paper, we consider both min-max and min-sum variants of the
quickest pair-visibility problem for a simple polygon. That is, either
we want to minimize the maximum length of two traveled paths (min-max)
or we want to minimize the sum of the lengths of two traveled paths
(min-sum).  We give a sweep-line-like approach that ``rotates'' the
lines-of-sight along vertices on the shortest path between the start
positions, allowing us to evaluate a linear number of candidate
solutions on these lines. Throughout the sweep, we encounter solutions
to both variants of the problem. We further show that our technique
can be implemented in linear time.
\begin{figure}[t]
  \centering \includegraphics[width=0.9\textwidth]{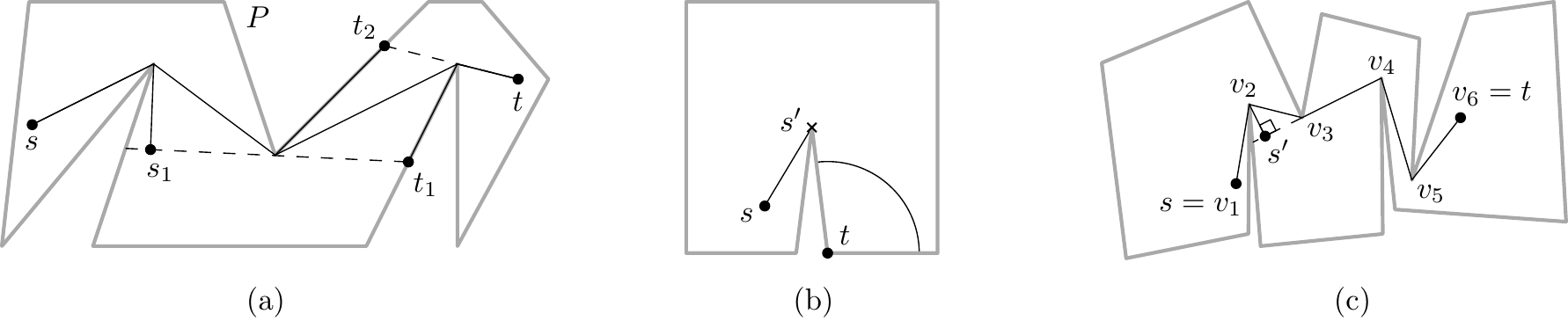}
  \caption{\small (a) The quickest pair-visibility problem finds two
    paths $\pi(s,s_1)$ and $\pi(t,t_1)$ such that
    $\overline{s_1t_1}\subset P$ and
    $\max\{|\pi(s,s_1)|,|\pi(t,t_1)|\}$ or $|\pi(s,s_1)|+|\pi(t,t_1)|$
    is minimized.  The quickest visibility problem for query point $t$
    finds a shortest $\pi(s,t_2)$ with $\overline{tt_2}\subset P$.
    (b) \textbf{min-max:} Every pair $(s',t^*)$, where $t^*$ is some
    point within the geodesic disk centered in $t$ with radius
    $\pi(s,s')$, is an optimal solution to the min-max
    problem. (c) \textbf{min-sum:} An instance where
      $|\pi(s,s')|+|\pi(t,v_4)|=|\pi(s,v_4)|+|\pi(t,v_5)|$. Therefore,
      both $(s',v_4)$ and $(v_4,v_5)$ are optimal solutions to the
      min-sum problem.}  
  \label{fig:problems}
\end{figure}
	
We also consider a query version of this problem for the min-max
variant. We can preprocess a simple $n$-gon in linear time so that the
minimum of the longer distance the two query points travel can be
computed in $O(\log^2 n)$ time for any two query points.
	
\section{Preliminaries}
Let $P$ be a simple polygon and \boundaryOfP be its boundary where
\boundaryOfP$\subset P$. The vertices of $P$ are given in
counter-clockwise order along \boundaryOfP.  We denote the shortest
path within $P$ between two points $p,q\in P$ by $\pi(p,q)$ and its
length by $|\pi(p,q)|$.  Likewise, we denote the shortest path within
$P$ between a point $p\in P$ and a line segment $\ell\subset P$ by
$\pi(p,\ell)$.  We say a point $p\in P$ is \emph{visible} from another
point $q\in P$ (and $q$ is visible from $p$) if and only if the line
segment $\overline{pq}$ is 
contained in $P$.
	
For two starting points $s$ and $t$, our task is to compute a pair
$(s',t')$ of points such that $s'$ and $t'$ are visible to each other,
where we wish to minimize the lengths of $\pi(s,s')$ and
$\pi(t,t')$. In the min-max setting, we wish to minimize
$\max\{|\pi(s,s')|,|\pi(t,t')|\}$. For the min-sum setting, we wish to
minimize $|\pi(s,s')| + |\pi(t,t')|$.  Note that, for both variants,
the optimum is not necessarily unique; see \fig~\ref{fig:problems}(b)
and (c).
	
We say a segment $g$ is \emph{tangent} to a path $\pi$ at a vertex $v$
if $v\in g\cap \pi$ and $v$'s neighboring vertices on $\pi$ are in a
closed half-plane bounded by the line containing $g$. Let
  $ \langle v_{0},v_{1},\ldots, v_{k-1},v_{k}\rangle$ be the sequence of vertices on
  $\pi(s,t)$ with $s=v_0$ and $t=v_k$.
	
\begin{lemma}
  \label{lemma:contains-vertex}
  Unless $s$ and $t$ are visible to each other, there is an optimal
  solution $(s^*,t^*)$ such that $\overline{s^*t^*}$ is tangent to the
  shortest path $\pi(s,t)$ at a vertex $v$ of $\pi(s,t)$.
\end{lemma}
\begin{proof}
  We first show that there is a vertex of $P$ lying on
  $\overline{s^*t^*}$.  Without loss of generality, assume that
  $\overline{s^*t^*}$ is
  horizontal with $s^*$ lying to the left of $t^*$.
  Let $\ell=\overline{xx'}$ be the maximal segment contained in $P$
  that contains $\overline{s^*t^*}$ with $x$ closer to $s^*$ than to
  $t^*$.  If $s=s^*$ (or $t=t^*$), then the lemma holds immediately
  because $s$ (or $t)$ is an endpoint of $\overline{s^*t^*}$.
  Assume to the contrary
  that $\overline{s^*t^*}$ contains no vertex of $P$.  Then there are
  points $p\in P$ in a neighborhood of $s^*$ and $q\in P$ in a
  neighborhood of $t^*$ such that $p$ and $q$ are visible to each other, and
  $\max\{|\pi(s,p)|, |\pi(t,q)|\}<\max\{|\pi(s,s^*)|,|\pi(t,t^*)|\}$
  and $|\pi(s,p)|+|\pi(t,q)|<|\pi(s,s^*)|+|\pi(t,t^*)|$.  This
  contradicts the optimality of $(s^*,t^*)$.

  We now show that $\overline{s^*t^*}$ contains a vertex of
  $\pi(s,t)$. 
  Let $s'$ be the vertex on $\pi(s,s^*)$ preceding $s^*$
  and let $t'$ be the vertex on $\pi(t,t^*)$ preceding $t^*$.
  Consider first the case that both $s'$ and $t'$ lie 
  below the line through $\ell$. See
  Figure~\ref{fig:ShortestPathVertex}(a).  Then \boundaryOfP touches
  $\overline{s^*t^*}$ at a vertex $v$ locally from below. Otherwise,
  $(s^*, t^*)$ is not optimal by the same argument as in the previous
  paragraph.  Then $s^*\in\overline{xv}$ and
  $t^*\in\overline{vx'}$. The path $\pi(s,t)$ 
  must cross $\overline{xv}$ at a point $y$ and $\overline{vx'}$ at a
  point $y'$. Since $y$ and $y'$ are visible to each
  other, and $\pi(s,t)$ is a shortest path, $\pi(s,t)$ contains
  $\overline{yy'}$, which in turn contains $v$. Thus $v$ lies on $\pi(s,t)$
  and $\overline{s^*t^*}$ is tangent to $\pi(s,t)$ at $v$.
	
  Consider now the case that $s'$ and $t'$ lie on different sides of
  the line through $\ell$. 
  Without loss of generality,
  assume that $s'$ lies below the line and $t'$ lies above the line.
  \ccheck{Then $\overline{s^*t^*}$
  	intersects $\pi(s,t)$. We first show that $\overline{s^*t^*}$
  	contains an edge of $\pi(s,t)$. Assume to the contrary that
  	$\overline{s^*t^*}$ intersects $\pi(s,t)$ only at a point, say $u$. 
  Then there is another line segment $\ell'\subset P$ 
  containing $u$ and intersecting both $\overline{s^*s'}$ and $\overline{t^*t'}$. 
  See Figure~\ref{fig:ShortestPathVertex}(b). 
  This contradicts that $(s^*,t^*)$ is an optimal solution because, for
  $s''=\ell'\cap \overline{s^*s'}$ and $t''= \ell'\cap\overline{t^*t'}$,
  $d(s,s'')<d(s,s^*)$ and/or $d(t,t'')<d(t,t^*)$ 
  and $s''$ and $t''$ are visible to each other. If $u$ coincides with $s^*$ 
  or $t^*$, only one of the distance inequalities above holds, we hence consider
  lexicographic smallest (max,min) solutions in the min-max setting to establish
  the contradiction. 
  Therefore, $\overline{s^*t^*}$
  contains an edge of $\pi(s,t)$, say $\overline{vv'}$. 
  Moreover, one of $v$ and $v'$ touches $\overline{s^*t^*}$ from above,
  and the other touches $\overline{s^*t^*}$ from below since 
  $s'$ and $t'$ are on different sides of $\ell$. }
  See Figure~\ref{fig:ShortestPathVertex}(c). 
  Thus, we can assume that $\boundaryOfP$ touches $\overline{s^*v}$
  at a vertex $v'$ locally from below. Then
  $\pi(s,t)$ must cross $\overline{xv'}$ at a point $y$, and $\overline{vx'}$ at a
  point $y'$. Since $y$ and $y'$ are visible to each
  other, and $\pi(s,t)$ is a shortest path, $\pi(s,t)$ contains
  $\overline{yy'}$, which in turn contains both $v'$ and $v$.  Thus both $v'$ and $v$ 
  lie on $\pi(s,t)$ and $\overline{s^*t^*}$ is tangent to $\pi(s,t)$ at both $v'$  and $v$.
\end{proof}
\begin{figure}[!t]
  \centering \includegraphics[width=0.8\textwidth]{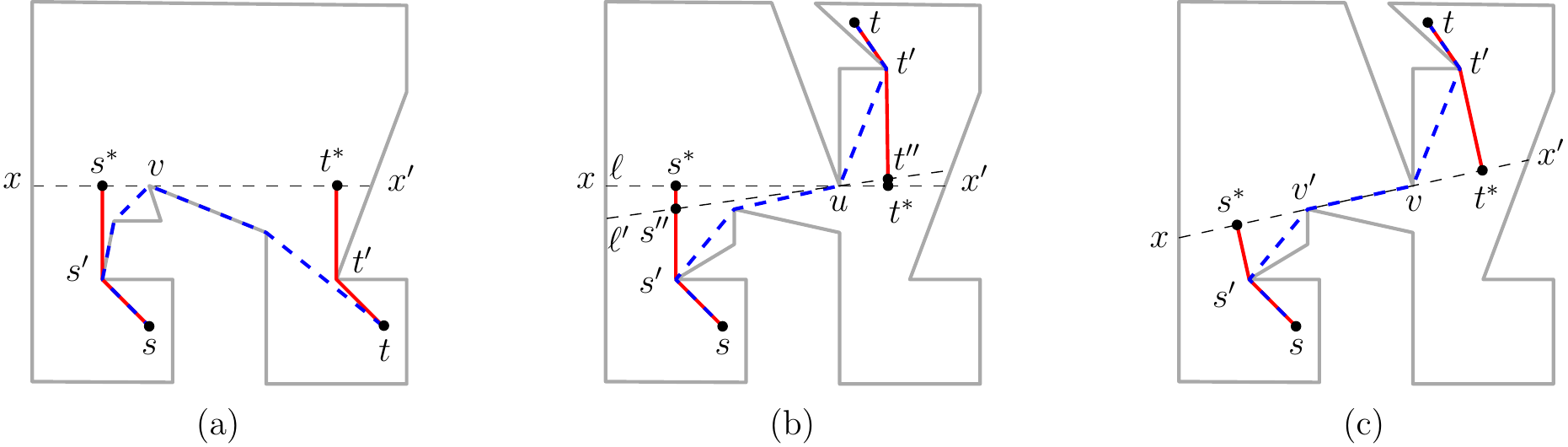}
  \caption{\small Illustrating cases for the proof of
    Lemma~\ref{lemma:contains-vertex}. The examples show the shortest
    (geodesic) path $\pi(s,t)$ (blue dashed) and the line $\ell$
    tangent to $v\in\pi(s,t)$ that, among all lines tangent to $v$
    minimizes shortest paths from $s,t$ to $\ell$ (red).  (a) Both $s'$
    and $t'$ lie on the same side of $\ell$ through $s^*$ and
    $t^*$. (b) $s'$ and $t'$ lie on different sides of $\ell$.
    (c) The shortest path $\pi(s,t)$ passes through $v$ and
    $v'$.}\label{fig:ShortestPathVertex}
\end{figure}
		
\section{Computing All Events for a Sweep-Line-Like
  Approach} \label{section:ComputingEvents} In the remaining part of
the paper, we use $(s^*, t^*)$ to denote the optimal solution pair
from $s$ and $t$ to a given line (and not necessarily a global optimal
solution for the quickest pair-visibility problem). For each vertex
$v$ on $\pi(s,t)$ we compute a finite collection of lines through $v$,
each being a configuration at which the combinatorial structure of the
shortest paths $\pi(s,s^*)$ and/or $\pi(t,t^*) $ changes.  To be more
precise, at these lines either the vertices of $\pi(s,s^*)$ or
$\pi(t,t^*)$ (except for $s^*$ and $t^*$) change or an edge of
\boundaryOfP changes that is intersected by the extension of
$\overline{s^*t^*}$. To explain how to compute these lines, we
introduce the concept of a \emph{line-of-sight}.
	
\begin{definition} [line-of-sight]
  We call a segment $\ell$ a \emph{line-of-sight} if (i)
    $\ell$ is a maximal segment contained in $P$, and
  (ii) $\ell$ is tangent to $\pi(s,t)$ at a vertex $v\in \pi(s,t)$.
\end{definition}
The algorithm we present is in many aspects similar to a sweep-line
strategy, except that we do not sweep over the scene in a standard
fashion but rotate a \emph{line-of-sight} $\ell$ in $P$ around the
vertices of the shortest path $\pi(s,t)$ in order from $s$ while making use
of Lemma~\ref{lemma:contains-vertex}. Recall that
  $ \langle v_{0},v_{1},\ldots, v_{k-1},v_{k}\rangle$ is the sequence of vertices on
  $\pi(s,t)$ with $s=v_0$ and $t=v_k$. The process will be
initialized with a line-of-sight that contains $s$ and $v_1$ and is
then rotated around $v_1$ (while remaining tangent to $v_1$) until it
hits $v_2$, see \fig~\ref{fig:BoundaryEvents}(a). In general, the
 line-of-sight is rotated around $v_{i}$ in a way so that it
remains tangent to $\pi(s,t)$ at $v_{i}$ (it is rotated in
the interior of $P$) until the line-of-sight contains $v_{i}$ and
$v_{i+1}$, then the process is iterated with $v_{i+1}$ as the new
rotation center. The process terminates as soon as the line-of-sight
contains $v_{k-1}$ and $v_k=t$.
	
While performing these rotations around the shortest path vertices, we
encounter all lines-of-sight.  As for a standard sweep-line approach,
we will compute and consider events at which the structure of a
solution changes: this is either because the interior vertices of
$\pi(s,s^*)$ or $\pi(t,t^*)$ change or because the line-of-sight
starts or ends at a different edge of \boundaryOfP.  These events will
be represented by points on \boundaryOfP{} \ccheck{ (actually, 
we consider the events as vertices on \boundaryOfP unless they are already vertices).}
Between two consecutive events, we compute the local minima of
the relevant distances for the variant at hand in constant time and
hence encounter all local minima eventually.
	
There are three event-types to distinguish:
\begin{enumerate*}
\item \textbf{Path-Events} are endpoints of lines-of-sight that
  contain two consecutive vertices of the shortest path $\pi(s,t)$.
  See \fig~\ref{fig:BoundaryEvents}(a).
\item \textbf{Boundary-Events} are endpoints of lines-of-sight that
  are tangent at a vertex of $\pi(s,t)$ and contain at least one
  vertex of $P\setminus \pi(s,t)$ (potentially as an endpoint). See
  \fig~\ref{fig:BoundaryEvents}(b).
\item \textbf{Bend-Events} are endpoints of lines-of-sight where the
  shortest path from $s$ (or $t$) to the line-of-sight gains or loses
  a vertex while rotating the line-of-sight around a vertex $v$. See
  \fig~\ref{fig:BoundaryEvents}(c). Note that bend-events can coincide
  with path- or boundary-events.
\end{enumerate*}
	
\begin{figure}[!t]
  \centering \includegraphics[width=0.8\textwidth]{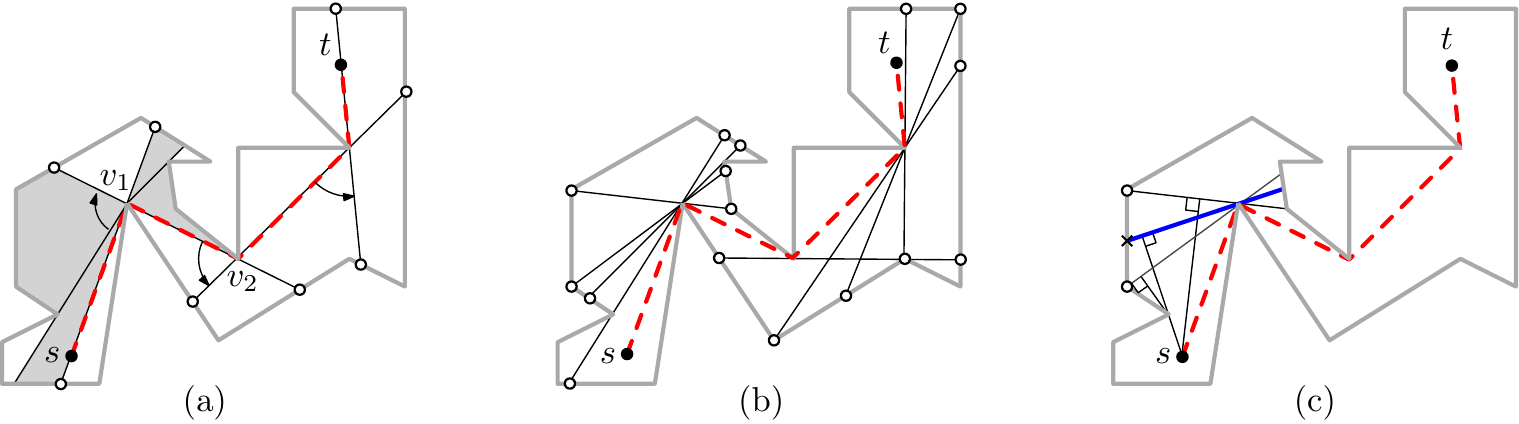}
  \caption{\small Path-, boundary-, and bend-events. (a) The endpoints
    of the line-of-sight through $\overline{sv_1}$ make up the first
    path-event. The line-of-sight rotates until it hits the next
    path-event: the endpoints of the line-of-sight through
    $\overline{v_1v_2}$. (b) Boundary-events that are not path-events. (c) A bend-event
    (marked with a cross) occurs between the two boundary-events. The
    shortest path from $s$ to the line-of-sight changes at the
    bend-event.}
  \label{fig:BoundaryEvents}
\end{figure}

We will need to explicitly know both endpoints of the line-of-sight on
\boundaryOfP at each event and the corresponding vertex of $\pi(s,t)$
on which we rotate.
\begin{lemma}[Computing path- and
  boundary-events]\label{lemma:InitQueue}
  For a simple polygon $P$ with $n$ vertices and points $s,t\in P$,
  the queue $\mathcal{Q}$ of all path- and boundary-events of the
  rotational sweep process, ordered according to the sequence in which
  the sweeping line-of-sight encounters them, can be initialized in
  $O(n)$ time.
\end{lemma}
\begin{proof}
  Consider some line-of-sight $\ell$ that is tangent to a vertex
  $v_{i}\in \pi(s,t)$ for some $0<i<k$.
  Then $\ell$ subdivides $P$ into a number of subpolygons.  Consider
  $\ell$ as the union of two (sub)segments $\ell^+$ and $\ell^-$ of
  $\ell$ induced by $v_i$
  such that $\ell^+\cap \ell^- = \{v_i\}$ and $\ell^-$ is incident to
  the subpolygon of $P$ induced by $\ell$ containing $s$.
	
  We will discuss the computation of all boundary- and path-events
  swept by $\ell^+$.  The other events swept by $\ell^-$ can be
  computed in a second round by changing the roles of $s$ and $t$.  We
  do not maintain a queue for the events explicitly; instead we will
  introduce new vertices on \boundaryOfP or label existing vertices of
  \boundaryOfP as events. Later the events will be considered by
  following two pointers to vertices on \boundaryOfP and hence by
  processing the vertices in the order of their occurrence
  on \boundaryOfP.
	
  We start with computing all path-events swept by $\ell^+$. For this
  we compute the \emph{shortest path map} $M_s$ of $s$ in $P$. The shortest
  path map of $s$ is a decomposition of $P$ in $O(n)$ triangular cells
  such that the shortest path from $s$ to any point within a cell is
  combinatorially the same. It can be obtained by extending every edge
  of the shortest path tree of $s$ towards its descendants until it
  reaches \boundaryOfP in linear time~\cite{Guibasetal1987}.  A
  path-event occurs when a line-of-sight contains two consecutive
  vertices of $\pi(s,t)$.
   Note that for each path-event, $\ell^+$ appears 
  as an edge of $M_s$ and its endpoints appear as vertices of $M_s$
  (see also Figure~\ref{fig:SPM}(a)).  For each index
    $i$ with $0< i< k$, we find the edge incident to $v_i$ and
  parallel to $\overline{v_{i-1}v_{i}}$ by considering every
  edge of $M_s$ incident to $v_i$. This takes $O(n)$ time in total
  since there are $O(n)$ edges of $M_s$ and we consider every edge at
  most once. Note that the path-event induced by $v_{k-1}$ and $t$ is
  an exception, but it can also be computed in $O(1)$ time during the process
  \ccheck{by considering the triangle of $M_s$ that contains $t$.}
  \begin{figure}
    \centering
    \includegraphics[width=0.5\textwidth]{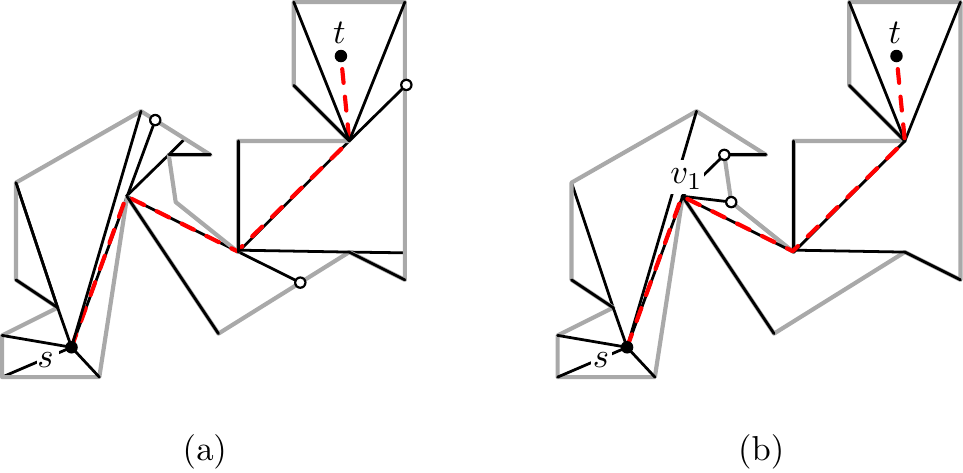}
    \caption{\small (a) The shortest path map $M_s$. 
        All path-events, swept by $\ell^+$, appear as endpoints of edges
        of $M_s$, except the one induced by $v_{k-1}$ and $t$; these events are marked with small disks. (b) The
        shortest path tree $T_s$. The boundary-events, swept by
        $\ell^+$, tangent to $\pi(s,t)$ at $v_1$ are marked. Clearly,
        the parent vertex of these vertices in $T_s$ is $v_1$.}
    \label{fig:SPM}
  \end{figure}
	
  For computing the boundary-events, we use the following properties.
  While rotating around $v_{i}$ from the position where $\ell$
  contains $v_{i-1}$ to the position in which $\ell$ contains
  $v_{i+1}$, let $A_i^+$ ($A_i^-$) be the region of $P$ that is swept
  over by $\ell^+$ ($\ell^-$). (See \fig~\ref{fig:queue_init}.)
  Observe that
  \begin{description*}
  \item[P1\label{prop:Aplus:disj}] all $A_i^+$ for $0<i<k$
    are pairwise disjoint in their interior,
  \item[P2\label{prop:Aminus:disj}] all $A_i^-$ for $0<i<k$
    are pairwise disjoint in their interior,
  \item[P3\label{prop:Aplus:path}] for all $0<i<k$ and all
    points $p\in A_i^+$ the shortest path $\pi(s,p)$
    contains $v_{i}$ (i.e., $v_i$ is the predecessor of $p$
      on $\pi(s,p)$),
  \item[P4\label{prop:Aminus:path}] for all $0<i<k$ and all
    points $p\in A_i^-$ the shortest path $\pi(p,t)$ contains $v_{i}$
    (i.e., $v_i$ is the successor of $p$ on $\pi(p,t)$).
  \end{description*}

  To compute all boundary-events that are vertices of $P$ swept by
  $\ell^+$, we will make use of the shortest path tree $T_s$ for $s$
  in $P$.  A boundary-event $x$ is defined by a vertex
  $v_{i}\in\pi(s,t)$ such that the line-of-sight that contains $x$
  (potentially as one endpoint) is tangent to $\pi(s,t)$ in $v_{i}$.
  It follows from \propthree, that $\overline{v_{i}x}$ is an edge of
  $T_s$ (and by that it cannot be obstructed
  by
  edges of $P$) and $x\notin \pi(s,t)$. So the vertices of $P$ whose
  parent vertex in $T_s$ is a vertex of $\pi(s,t)$ are possible
  boundary-events. In order to compute all boundary-events we consider
  all consecutive path-events and compute all corresponding
  boundary-events by following \boundaryOfP and checking the vertices
  within the candidate set (see
    Figure~\ref{fig:SPM}(b)).
  We compute the boundary-events that are vertices of $P$ swept by
  $\ell^-$ in a similar way.
	
  So far we have labeled all vertices $x$ on \boundaryOfP that are
  boundary-events. We still need to compute the other endpoint
  $\tilde{x}$ of the line-of-sight $\overline{x\tilde{x}}$ that is
  tangent in $v_i$.  Let $\overline{x_i\tilde{x}_i}$ be the
  line-of-sight at the path-event $x_i$ so that
  $\tilde{x}_i, v_{i-1}, v_i, x_i \in \ell$. See
  \fig~\ref{fig:queue_init}. While rotating $\ell$ around $v_i$,
  $\ell^+$ sweeps over $A_i^+$ until the next path-event is met.  Let
  $E_i^+$ be the sequence of the path- and boundary-events in $A^+_i$
  we obtained so far sorted in counter-clockwise order along
  \boundaryOfP.  The order of events in $E^+_i$ is the same as the
  order in which $\ell^+$ sweeps over them.
  Our goal is to compute $\tilde{x}$ for every event in $E^+_i$ in
  order.  To do this, we consider the (triangular) cells of the shortest path
  map $M_t$ of $t$
  incident to $v_i$ one by one in counter-clockwise order around $v_i$
  starting from the cell incident to $\tilde{x}_i$. Since
  every point in such cells is visible from $v_i$, we can determine if
  $\tilde{x}$ is contained in a cell in constant time for any event
  $x\in E^+_i$.  Therefore, we can compute $\tilde{x}$ for every event
  $x$ in $E^+_i$ in time linear in the number of the cells of $M_t$
  incident to $v_i$ and the number of events of $E^+_i$, giving us all
  path- and boundary-events in $O(n)$ total time.
  \begin{figure}[!tb]
    \begin{center}
      \includegraphics[scale=0.8]{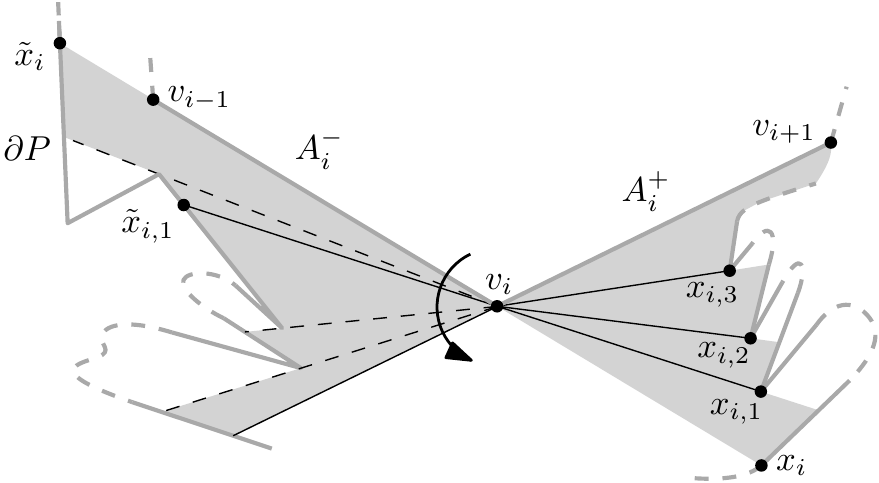}
      \caption{\small Let
        $E^+_i = \langle x_{i,1},\ldots,x_{i,j}\rangle$ for an index
        $1\leq j\leq n$.  We start at $\tilde{x}_i$ and follow the
        (triangular) cells of $M_t$ incident to $v_i$ in
        counter-clockwise order around $v_i$ until we find
        $\tilde{x}_{i,1}$. Then we continue to follow such cells until
        we find $\tilde{x}_{i,2}$, and so on.}
      \label{fig:queue_init}
    \end{center}
  \end{figure}
\end{proof}
        
Once we initialized the event queue $\mathcal{Q}$, we can now compute
and process bend-events as we proceed in our line-of-sight rotations.

\begin{lemma}\label{lemma:BendEvents}
  All bend-events can be computed in $O(n)$ time, sorted in
    the order as they appear on the boundary of $P$.
\end{lemma}
\begin{proof} Bend-events occur between consecutive path- and
    boundary-events; they can also coincide with these events.
  We assume that all path- and boundary-events are already
  computed. Additionally, we assume that all vertices of the boundary-
  and path-events (the endpoints of the corresponding lines-of-sight)
  are inserted on \boundaryOfP. Recall that, for each event, we know
  both endpoints of the line-of-sight $\ell$ on \boundaryOfP and the
  corresponding vertex of $\pi(s,t)$ on which we rotate. The
    path- and boundary-events define the area which is swept over by
    $\ell$. Thus, we know which positions for $\ell$ we have to
    consider in order to compute all bend-events.

  As in the proof of Lemma~\ref{lemma:InitQueue}, we consider the
  line-of-sight $\ell$ tangent to a vertex $v\in \pi(s,t)$ as
  the union of two (sub)segments $\ell^+$ and $\ell^-$ of
    $\ell$ induced by $v$ such that $\ell^+\cap \ell^- = \{v\}$ and
    $\ell^-$ is incident to the subpolygon of $P$ induced by $\ell$
    containing $s$. We discuss the computation of all bend-events
  that are encountered by $\ell^-$.  The bend-events that are swept
  over by $\ell^+$ can be computed in a second round by changing the
  roles of $s$ and $t$.
 
  We start with the path-event defined by $s$ and $v_1$, and consider
  all events in the order they appear.  Let $\ell$ be the 
  line-of-sight rotating around a vertex $v$ and denote by $x$ the
  endpoint of $\ell^-$ other than $v$.  To find the bend-events
  efficiently, we compute and maintain the shortest path $\pi(s,\ell)=\pi(s,\ell^-)$
  over the events.
 
  \newcommand{\typeone}{type \textbf{T1}} \newcommand{\typetwo}{type
    \textbf{T2}}
 
  While $\ell$ rotates around $v$, the combinatorial structure of
  $\pi(s,\ell)$ may change. Specifically, let $e_\ell=\overline{uw}$
    denote the edge of $\pi(s,\ell)$ incident to $\ell$ with $w$ on
    $\ell$.
  Note that during the rotation of $\ell$, all the edges of
  $\pi(s,\ell)$ are stationary, except that $e_\ell$ rotates around
  $u$. Therefore, a change in the combinatorial structure of
  $\pi(s,\ell)$ occurs only when $e_\ell$ hits a vertex $u'$ of $P$
  (if $u'$ at this event is an endpoint of $e_\ell$, then this bend-event
  coincides with a previously computed boundary-event)
  and splits into two edges sharing $u'$ (an event of \typeone{}) or
  the two edges of $\pi(s,\ell)$ incident to $u$ become parallel (an
  event of \typetwo{}). (Then they merge into one and $u$ disappears
  from $\pi(s,\ell)$.)  See
  Figure~\ref{fig:BendEvents}. From any event of the two event types
  above, $e_\ell, u$, and $\pi(s,\ell)$ are updated accordingly.
  Additionally, $x$ is updated and its new position is inserted as a
  vertex on \boundaryOfP as it represents a bend-event.
  \begin{lemma}
    An event of \typeone{} occurs only when (i) $x$ reaches a
      vertex $u'$, or (ii) $e_\ell$ hits a vertex $u'$ of $\pi(s,t)$
    in its interior.  Moreover, for case (ii), $u$ and $u'$ are
    consecutive in $\pi(s,t)$.
  \end{lemma}
  \begin{proof}
    Imagine $\ell$ is rotated around $v$ infinitesimally further 
    from the current event. Then either $e_\ell$ is orthogonal to $\ell$
    or not.
    If $e_\ell$ is not orthogonal to $\ell$,
    the closest point in $\ell$ from $s$ is
    $x$. 
    Thus, the only way that $e_\ell$ hits a vertex of $P$ is that $x$
    reaches $u'$.  See Figure~\ref{fig:BendEvents}(a).
   
    Now consider the case that $e_\ell$ is orthogonal to
    $\ell$. Notice that the shortest path from a vertex $v$
      to a segment within a simple polygon lies inside a
      \emph{funnel}, a region bounded by the shortest paths
      from $v$ to both endpoints of the segment and the segment. For more details
      see~\cite{Guibasetal1987}.  Thus, $u'$ is contained in
    $\pi(u,v)$.  See Figure~\ref{fig:BendEvents}(b).  Since $\pi(u,v)$
    is a subpath of $\pi(s,t)$, $u'$ is a vertex of
    $\pi(s,t)$, and thus $u$ is the vertex of $\pi(s,t)$ previous to
    $u'$ from $s$.
  \end{proof}

 \begin{figure}[t]
   \centering \includegraphics[width=0.9\textwidth]{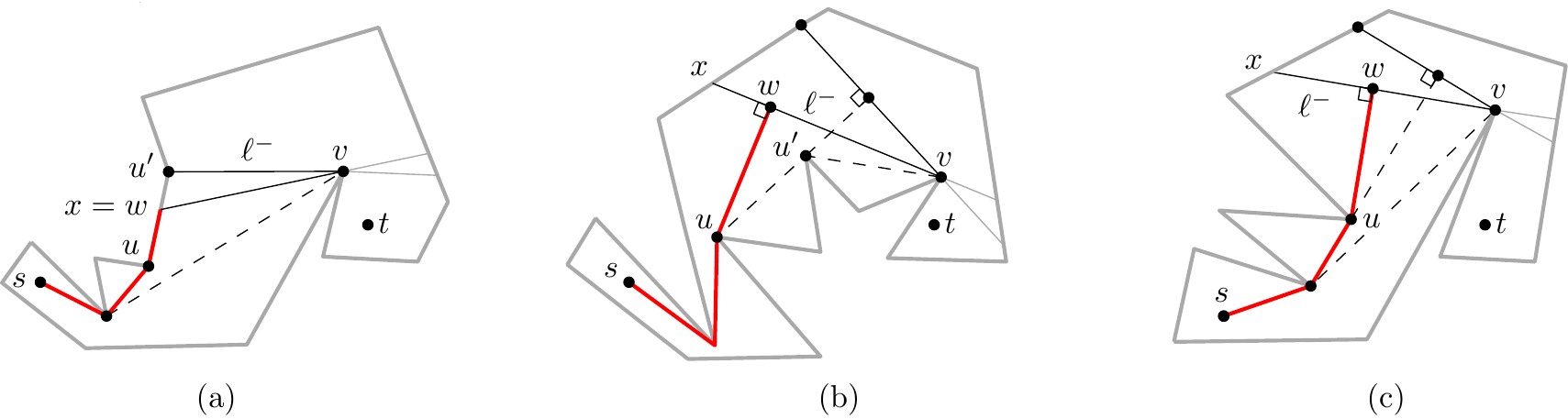}
   \caption{\small (a) A bend-event of \typeone{} occurs when
     $x=u_\ell$ reaches $u'$. (b) A bend-event of \typeone{}
     occurs when $e_\ell=\overline{uw}$ hits a
     vertex $u'$ of $\pi(s,t)$. (c) A bend-event of
     \typetwo{} occurs when two edges incident to $u$
     are parallel.}
   \label{fig:BendEvents}
 \end{figure}
 \begin{lemma}\label{lemma:vertices-shortestpath}
   Once a vertex disappears from $\pi(s,\ell)$, it never appears again
   on the shortest path during the rotation of the line-of-sight $\ell$.
 \end{lemma}	
 \begin{proof}
   Assume to the contrary that there is a vertex $u$ that disappears
   from $\pi(s,\ell_1)$, but then appears again on $\pi(s,\ell_2)$ for
   two distinct lines-of-sight $\ell_1$ and $\ell_2$ during the rotation.
   \ccheck{First note that if $u$ is an endpoint of $\pi(s,\ell_1)$ (or $\pi(s,\ell_2)$), it is 
   a boundary- and bend-event, and would only appear once
   when rotating the line-of-sight.
   Therefore, both $\pi(s,\ell_1)$ and $\pi(s,\ell_2)$ must contain $u$ in their
   interiors, and both of them also contain
   $\pi(s,u)$ in their interiors.} Since $u$ disappears from $\pi(s,\ell_1)$, 
   the edge of
   $\pi(s,\ell_1)$ incident to $u$ (on $\pi(u,\ell_1)$) is orthogonal to
   $\ell_1$. We claim that $u$ appears on $\pi(s,\ell_2)$ due
     to case (ii) of \typeone{}, that is, the edge of $\pi(s,\ell_2)$
   incident to $\ell_2$ hits $u$. Assume to the contrary that $u$
   appears on $\pi(s,\ell_2)$ due to case (i) of \typeone{}.
   However, $u$ (and its event vertex on \boundaryOfP) is
   already swept by a line-of-sight before we consider $\ell_2$
   because it appears on $\pi(s,\ell_1)$. \ccheck{By \proptwo, $\ell^-$ sweeps a vertex only once.} Thus, $u$ appears on
   $\pi(s,\ell_2)$ due to case (ii) of \typeone{},
   and the edge of $\pi(s,\ell_2)$ incident to $u$ is orthogonal to
   $\ell_2$. This means that $\ell_1$ and $\ell_2$ are parallel.
		
   Since $\ell_1$ and $\ell_2$ are parallel, they are tangent to
   $\pi(s,t)$ at two distinct vertices, say $u_1$ and $u_2$,
   respectively. \ccheck{Without loss of generality, assume that $u_1$ is closer to
   	$s$ than $u_2$.} We show that $\pi(p_1,p_2)$ contains $u_1$ for
   any two points $p_1 \in P_1$ and $p_2\in \ell_2$, where
   $P_1$ is the subpolygon bounded by $\ell^-_1$ containing
     $s$. 
   \ccheck{Since both $u_1$ and $u_2$ are vertices of $\pi(s,t)$,
   $\pi(s,u_2)$ contains $u_1$. Let $p$ be the point on $\ell^-_2$ farthest from $u_2$
	such that $\pi(s,p)$ contains $u_1$. Since the boundary of $P$ intersect neither 
	$\overline{u_1p}$ nor $\overline{u_2p}$, $\pi(u_1,u_2)$ is contained in the
	triangle with corners $u_1, u_2, p$. No line segment parallel to $\ell_2$ is tangent to $\pi(s,t)$ at $u_1$, which is a contradiction. Therefore, $\pi(s,p_2)$ contains $u_1$ for any point $p_2\in\ell_2$. Then since $\ell_1$ is tangent to $\pi(s,t)$,
	$\pi(p_1,p_2)$ contains $u_1$ for any two points $p_1 \in P_1$ and $p_2\in \ell_2$.}     
   Thus, $\pi(s,\ell_2)$ contains $\pi(s,u_1)$, and no vertex
   in $P_1$ other than the vertices of $\pi(s,u_1)$ appears on
   $\pi(s,\ell_2)$.  Since $u$ is contained in $P_1$, it cannot appear
   on $\pi(s,\ell_2)$, which is a contradiction.
 \end{proof}

 Using the two lemmas, we can compute all bend-events as follows.
 For a line-of-sight $\ell$ rotating around a vertex $v$, 
 we have three candidates for the
 next bend-event. Let $e$ be the edge of $P$ containing the endpoint of $\ell^-$
 other than $v$, and let $u'$ be the neighboring vertex of $u$ in $\pi(u,t)$. The next bend-event is (1) the endpoint of $e$ not contained in $\pi(s,\ell)$ if it exists, (2)
 the intersection point between $e$ and the line through $v$ and orthogonal
 to $uu'$ if it exists,  or
 (3) the intersection point between $e$ and the line
 through $v$ and orthogonal to $u''$ if it exists, where $u''$ is the neighboring vertex 
 of $u$ in $\pi(s,\ell)$ closer to $s$. Note that the first two cases are  \typeone{} events and the last case is a \typetwo{} event. 
 We can compute all of the three events in constant time. Also, we can update $u, e_\ell,x$ and $\pi(s,\ell)$ 
 accordingly in constant time. 
 Therefore, the time for computing all bend-events is linear in the 
 amount of the combinatorial change on $\pi(s,\ell)$. 
 By Lemma~\ref{lemma:vertices-shortestpath}, the the amount of the combinatorial change is $O(n)$, and therefore,
 we can compute all bend-events in $O(n)$ time. 
\end{proof}

\section{Algorithm Based on a Sweep-Line-Like Approach}\label{section:Algorithm}

In this section, we present a linear-time algorithm for computing the
minimum distance that two points $s$ and $t$ in a simple polygon $P$
travel in order to see each order.  We compute all events defined in
Section~\ref{section:ComputingEvents} in linear time.  The remaining
task is to handle the lines-of-sight lying between two consecutive
events.

\begin{lemma}\label{lemma:greenstarformula}
  For any two consecutive events, the line-of-sight $\ell$ lying
  between them that minimizes the sum or the maximum of the distances 
  from $s$ and $t$ to $\ell$ can be found in constant time.
\end{lemma}
\begin{proof}
  Let $\mathcal{L}$ be the set of all lines-of-sight lying between the
  two consecutive events. 
  We assume that $\mathcal{L}$ contains no vertical line-of-sight. 
  Otherwise, we consider the subset containing all lines-of-sight 
  with positive slopes, and then the subset containing
  all lines-of-sight with negative slopes.
  
  These lines-of-sight share a common vertex $v$ of $\pi(s,t)$.  
  We will give an algebraic function for $|\pi(s,\ell)|$ for $\ell\in\mathcal{L}$.
  An algebraic function for $|\pi(t,\ell)|$ can be obtained analogously.
  Observe that $\pi(s,u)$ is the same for all $\ell\in\mathcal{L}$, where
  $u$ is the second to the last vertex $u$ of $\pi(s,\ell)$ from $s$.
  Thus, we consider only the length of $\pi(u,\ell)$, which is 
  a line segment.
  The length is either the Euclidean distance between
  $u$ and the line containing $\ell$, or the Euclidean distance
  between $u$ and the endpoint of $\ell$ closest to $u$.
  We show how to handle the first case only because the second case can be handled
  analogously.
  
  Let $\ell(\alpha)$ denote the line of
  slope $\alpha$ passing through $v$ for $\alpha>0$, which is represented as 
  $y= \alpha x+ f(\alpha)$, where $f(\alpha)$ is a function linear in
  $\alpha$.  
  Then the distance between $u$ and $\ell(\alpha)$ can be
  represented as 
  ${|c_1\alpha + c_2|}/{\sqrt{\alpha^2 + 1}}$, where $c_1$ and $c_2$
  are constants depending only on $v$ and $u$.
  Thus, our problem reduces to finding a minimum of the function 
  of the form
  ${(|c_1\alpha + c_2|+|c_1'\alpha + c_2'|)}/{\sqrt{\alpha^2 + 1}}$ and $\max(|c_1\alpha + c_2|,|c_1'\alpha + c_2'|)/{\sqrt{\alpha^2 + 1}}$, respectively,
  for four constants $c_1, c_2, c_1'$ and $c_2'$, and for all $\alpha$
  such that $\ell(\alpha)$ contains a line-of-sight in
  $\mathcal{L}$.
  We can find a
  minimum in constant time using elementary analysis.
\end{proof}

\begin{theorem}\label{theorem:main_result}
  Given a simple $n$-gon $P$ with no holes and two points $s,t\in P$,
  a point-pair $(s^*,t^*)$ such that (i) $\overline{s^*t^*}\subset P$
  and (ii) either $|\pi(s,s^*)|+|\pi(t,t^*)|$ or
  $max\{|\pi(s,s^*)|, |\pi(t,t^*)|\}$
  is minimized can be computed in $O(n)$ time.
\end{theorem}
\begin{proof}
  Our algorithm first computes all path- and boundary-events as
  described in Lemma~\ref{lemma:InitQueue}. The number of events
  introduced during this phase is bounded by the number of vertices of
  the shortest path maps, $M_s$ and $M_t$, respectively, which are
  $O(n)$.  In the next step, it computes the bend-events on
  \boundaryOfP as described in Lemma~\ref{lemma:BendEvents}, which can
  be done in $O(n)$ time.  Finally, our algorithm traverses
    the sequence of events.  Between any two consecutive events, it
    computes the respective local optimum in constant time by
    Lemma~\ref{lemma:greenstarformula}.
    It maintains the smallest one among the local optima computed so
    far, and returns it once all events are processed. Therefore the
    running time of the algorithm is $O(n)$.
  
  For the correctness, consider the combinatorial structure of
    a solution and how it changes. The path-events ensure that all
    vertices of $\pi(s,t)$ are considered as being the vertex lying on
    the segment connecting the solution $(s^*,t^*)$
    (Lemma~\ref{lemma:contains-vertex}). While the
    line-of-sight rotates around one fixed vertex of $\pi(s,t)$,
    either the endpoints of line-of-sight sweep over or become tangent
    to a vertex of \boundaryOfP. These are exactly the
    boundary-events. Or the combinatorial structure of $\pi(s,s^*)$ or
    $\pi(t,t^*)$ changes as interior vertices of $\pi(s,s^*)$ or
    $\pi(t,t^*)$ appear or disappear. These happen exactly at
    bend-events. Therefore, our algorithm returns an optimal
    point-pair.
\end{proof}

\begin{corollary}\label{corollary:modificationsToObjective}
  By the same algorithm, one can also compute optimal pairs
  $(s^*,t^*)$ that minimize
  \begin{itemize*}
  \item $\max(\lambda|\pi(s, s^*)|,(1-\lambda)|\pi(t, t^*)|)$ for some
    $0\leq \lambda\leq 1$,
  \item $\max(\alpha+|\pi(s, s^*)|,\beta+|\pi(t, t^*)|)$ for some
    $\alpha, \beta \in \mathbb{R}^+$.
  \end{itemize*}
\end{corollary}
The first modification introduced in
Corollary~\ref{corollary:modificationsToObjective} models that
Romeo and Juliet travel with different speeds. It is easy to
see, that this formulation is equivalent to minimizing the objective
$\max(\alpha|\pi(s, s^*)|,\beta|\pi(t, t^*)|)$ for some
$\alpha, \beta \in \mathbb{R}^+$. The second variant can be motivated
as follows: Imagine Romeo (and Juliet) is driving a car that
before departing from $s$ (and $t$) already drove a distance of $\alpha$
(and $\beta$).  The objective
$\max(\alpha+|\pi(s, s^*)|,\beta+|\pi(t, t^*)|)$ minimizes the largest
distance any of the two cars had to drive in order to establish a
line-of-sight.

\section{Quickest Pair-Visibility Query Problem}
In this section, we consider a query version of the min-max variant of
the quickest pair-visibility problem: Preprocess a simple $n$-gon $P$
so that the minimum traveling distance for two query points $s$ and
$t$ to see each other can be computed efficiently.  We can preprocess
a simple $n$-gon in linear time and answer a query in $O(\log^2n)$
time by combining the approach in Section~\ref{section:Algorithm} with
the data structure given by Guibas and
Hershberger~\cite{Guibas1989,Hersh-shortest-1991}.  For any two query
points $s$ and $t$ in $P$, the query algorithm for their data
structure returns $\pi(s,t)$, represented as a binary tree of height
$O(\log n)$, in $O(\log n)$ time~\cite{Hersh-shortest-1991}. Thus, we
can apply a binary search on the vertices (or the edges) on $\pi(s,t)$
efficiently.

Imagine that we rotate a line-of-sight along the vertices of
$\pi(s,t)$ for two query points $s$ and $t$ in $P$.
Lemma~\ref{lemma:contains-vertex} implies that there is a
line-of-sight containing $s^*$ and $t^*$, where $(s^*,t^*)$ is an
optimal solution. We call it an \emph{optimal line-of-sight}.  We
define the order of any two lines-of-sight as the order in which they
appear during this rotational sweep process. 
By the following lemma, we
can apply a binary search on the sequence of events along \boundaryOfP
and find two consecutive events
such that the respective local optimum achieved between them
  is a global optimal solution.

\begin{lemma}
  The geodesic distance between $s$ (and $t$) and the rotating
    line-of-sight increases (and decreases) monotonically as the
    line-of-sight rotates along the vertices of $\pi(s,t)$ from $s$.
\end{lemma}
\begin{proof}
  Let $\ell$ be a line-of-sight that is tangent to $\pi(s,t)$
    at a vertex $v$. Consider the subdivision of $P$ induced by $\ell$
    and let $P_s$ be the subpolygon that contains $s$.
    Let $\ell'$ be a line-of-sight that comes after $\ell$ during the
    rotational sweep process.  We claim that $\ell'$ does not
    intersect the interior of $P_s$.  If $\ell'$ is tangent to
    $\pi(s,t)$ at $v$, it never intersects the interior of $P_s$ as
    shown in the proof of Lemma~\ref{lemma:InitQueue}.  Assume that
    $\ell'$ is tangent to $\pi(s,t)$ at a vertex $u$ that comes after
    $v$ along $\pi(s,t)$ from $s$, but intersects the interior of
    $P_s$. Without loss of generality, assume that $\ell$ is
    horizontal and $P_s$ lies locally below $\ell$.  Then $u$ must lie
    strictly above the line containing $\ell$.  However, since both
    $v$ and $u$ are vertices of $\pi(s,t)$ and $\ell$ is tangent to
    $\pi(s,t)$ at $v$, there must be another vertex $u'$ of $\pi(s,t)$
    that lies on or below the line containing $\ell$ and appears
    between $v$ and $u$ along $\pi(s,t)$. See
    Figure~\ref{fig:Lemma10}.
		
    \begin{figure}[!t]
      \centering
      \includegraphics[width=0.25\textwidth]{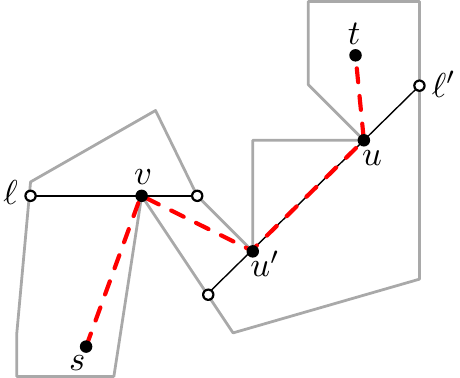}
      \caption{\small Let $\ell$ be a be a
        line-of-sight which is tangent to
        $\pi(s,t)$ at a vertex $v$. And let
        $\ell'$ be be a line-of-sight that comes
        after $\ell$ during the rotational sweep
        process.  Clearly,
        $|\pi(s,\ell')|\geq|\pi(s,\ell)|$.}
      \label{fig:Lemma10}
    \end{figure}
    
    Thus, $u$ is not visible from any point on $\ell$, and $\ell'$
    does not intersect the interior of $P_s$.  Since $\pi(s,\ell')$
    intersects $\ell$, we have $|\pi(s,\ell')|\geq|\pi(s,\ell)|$. The
    claim for $t$ and the rotating line-of-sight can be shown
    analogously.
\end{proof}

\subsection{Binary Search for the Path-Events}
We first consider the path-events, and find two consecutive
path-events containing an optimal line-of-sight between them.  Let
  $ \langle v_{0},v_{1},\ldots, v_{k-1},v_{k}\rangle$ be the sequence of vertices on
  $\pi(s,t)$ with $s=v_0$ and $t=v_k$. Due to the
shortest-path data structure by Guibas and Hershberger, we can obtain
$\pi(s,t)$ represented as a binary tree of height $O(\log n)$ in
$O(\log n)$ time.  Consider an edge $\overline{v_iv_{i+1}}$ of
$\pi(s,t)$. We can determine whether or not an optimal line-of-sight
is tangent to $\pi(s,t)$ at a vertex lying after $v_i$ along
$\pi(s,t)$ in $O(\log n)$ time. To do this, we compute the
line-of-sight $\ell$ containing $\overline{v_iv_{i+1}}$ in $O(\log n)$
time. We use the data structure for ray shooting given by Hershberger and Suri~\cite{Hershberger1995} with linear preprocessing and logarithmic query time. Then, we compute the length of $\pi(s,\ell)$
and $\pi(t,\ell)$ in $O(\log n)$ time \ccheck{using the data structure given by Guibas and Hershberger for computing the distance between a query point and a query line segment in $O(\log n)$ time~\cite{Guibas1989}.} An optimal
line-of-sight is tangent to $\pi(s,t)$ at a vertex lying
after $v_i$ if and only if $\pi(s,\ell)$ is shorter than
$\pi(t,\ell)$. Therefore, we can compute the two consecutive
path-events with an optimal solution lying between them in
$O(\log^2 n)$ time.

\subsection{Binary Search for the Boundary-Events}
Now we have the vertex $v_i$ of $\pi(s,t)$ contained in an optimal
line-of-sight.  We find two consecutive boundary-events
  defined by lines-of-sight tangent to $\pi(s,t)$ at $v_i$ such that
an optimal line-of-sight lies between them.  Let $\tilde{x}_i$ and
$x_i$ be the first points of \boundaryOfP hit by the rays from any
point in $\overline{v_{i-1}v_{i}}$ towards $v_{i-1}$ and $v_i$,
respectively. See Figure~\ref{fig:queue_init}. Similarly, let
$\tilde{x}_{i+1}$ and $x_{i+1}$ be the first points of \boundaryOfP
hit by the rays from any point in $\overline{v_{i}v_{i+1}}$ towards
$v_{i}$ and $v_{i+1}$, respectively. These four points of
  \boundaryOfP can be found in $O(\log n)$ time by the ray-shooting
  data structure~\cite{Hershberger1995}.  Without loss of generality,
we assume that a line-of-sight rotates around $v_i$ in the
counter-clockwise direction in the rotational sweep process.  Let
$\tilde{\gamma}$ be the part of \boundaryOfP lying between $\tilde{x}_i$
and $\tilde{x}_{i+1}$ in counter-clockwise order, and $\gamma$ be the
part of \boundaryOfP lying between $x_i$ and $x_{i+1}$ in
counter-clockwise order.  An optimal line-of-sight $\ell^*$ has one
endpoint on $\tilde{\gamma}$ and the other endpoint on $\gamma$.

We first find the edge of $\tilde{\gamma}$ (resp. $\gamma$) containing
an endpoint of $\ell^*$ by applying a binary search on the vertices of
$\tilde{\gamma}$ (resp. $\gamma$).  This gives two consecutive
boundary-events such that $\ell^*$ lies between them.  We now show how
to find the edge of $\gamma$ containing an endpoint of $\ell^*$.
The edge on $\tilde{\gamma}$ can be found analogously.

We perform a binary search on the vertices in $\gamma$ as follows.
Let $x^*$ be the endpoint of $\ell^*$ contained in $\gamma$.  For any
vertex $u$ of $\gamma$, we can determine which part of $\gamma$ with
respect to $u$ contains $x^*$ in $O(\log n)$ time.  To do
  this, we consider the line-of-sight $\ell$ containing the edge of
  $\pi(v_i,u)$ incident to $v_i$.  Observe that $\ell$ intersects
  $\pi(v_i,u)$ only in the edge including its endpoints as
  $\pi(v_i,u)$ is a shortest path.
See Figure~\ref{fig:query-bend}(a).  Since we can obtain the edge of
$\pi(v_i,u)$ incident to $v_i$ in $O(\log n)$ time using the
shortest-path data structure, we can obtain $\ell$ in the same
time. Here, to obtain the endpoint of $\ell$ on $\gamma$, we
use the ray-shooting data structure that supports $O(\log n)$ query
time~\cite{Hershberger1995}.  Then we compare $|\pi(s,\ell)|$ and
$|\pi(t,\ell)|$ in $O(\log n)$ time.  The point $x^*$ comes after $u$
from $x_i$ if and only if $|\pi(s,\ell)|<|\pi(t,\ell)|$.  Therefore,
we can determine which part of $\gamma$ with respect to $u$ contains
$x^*$ in $O(\log n)$ time, and thus the binary search is completed in
$O(\log^2 n)$ time.  In this way, we can compute two consecutive
boundary-events such that an optimal line-of-sight lies between them
in $O(\log^2 n)$ time.

\subsection{Binary Search for the Bend-Events}
Now we have two consecutive events in the sequence of all path- and
boundary-events that contain an optimal line-of-sight $\ell^*$ between
them.  Let $\ell_1$ and $\ell_2$ be two lines-of-sight corresponding
to the two consecutive events such that $\ell_2$ comes after $\ell_1$.
The remaining task is to handle the bend-events lying between them.
For the bend-events, we perform a binary search on the edges of
$\pi(s,\ell_1)\cup\pi(s,\ell_2)$ in $O(\log^2 n)$ time.  Then we
perform a binary search on the edges of
$\pi(t,\ell_1)\cup\pi(t,\ell_2)$ in $O(\log^2 n)$ time.  In the
following, we describe the binary search on
$\pi(s,\ell_1)\cup\pi(s,\ell_2)$.  The other one can be done
analogously.

We find the point $s'$ such that $\pi(s,s')$ is the maximal
  common subpath of $\pi(s,\ell_1)$ and $\pi(s,\ell_2)$ from $s$ in
$O(\log n)$ time using the shortest-path data
structure~\cite{Hersh-shortest-1991}. See
Figure~\ref{fig:query-bend}(b).  Then we obtain
$\pi'=\pi(s',\ell_1)\cup\pi(s',\ell_2)$ represented as a binary tree
of height $O(\log n)$ in $O(\log n)$ time.  Notice that
  $\pi'$ is a path from $\ell_1$ to $\ell_2$, concatenating the two
  shortest paths from $\ell_1$ to $s'$ and from $s'$ to $\ell_2$.

For an edge $e$ of $\pi'$, 
we use $\ell(e)$ to denote the line-of-sight containing $v_i$ and
orthogonal to the line containing $e$.
Observe that $\ell(e)$ comes after $\ell(e')$ if and only if $e$ comes
after $e'$ along $\pi'$ 
from $\ell_1$ (because the order of the edges of
  $\pi'$, as they appear on the path, are radially sorted around
  $v_i$). Also, given an edge $e$ of
$\pi'$, 
we can compute
$\ell(e)$ in constant time.  Using these properties, we can find two
consecutive edges $e$ and $e'$ of
$\pi'$ 
such that $\ell^*$ lies between $\ell(e)$ and $\ell(e')$ in $O(\log^2
n)$ time by applying a binary search on
$\pi'$ 
as we did for path- and boundary-events.

Now we have two consecutive events in the sequence of all
  path-, boundary- and bend-events that contain
  $\ell^*$ between them.  Recall that the combinatorial structure of
$\pi(s,\ell)$ (and
$\pi(t,\ell)$) is the same for every line-of-sight lying between the
two events. Let $(u_s,w_s)$ and
  $(u_t,w_t)$ be the edges of $\pi(s,\ell)$ and
  $\pi(t,\ell)$ incident to $\ell$ at $w_s$ and
  $w_t$, respectively, for any line-of-sight
  $\ell$ lying between the two events.  Using the shortest-path data
structure, we can obtain $u_s, u_t, |\pi(s,u_s)|$ and
$|\pi(t,u_t)|$ in $O(\log
n)$ time.  Then we apply the algorithm in
Lemma~\ref{lemma:greenstarformula} to find an optimal line-of-sight in
constant time. In this way, we can obtain an optimal line-of-sight in
$O(\log^2n)$ time in total.

\begin{figure}
  \centering
  \includegraphics[width=0.8\textwidth]{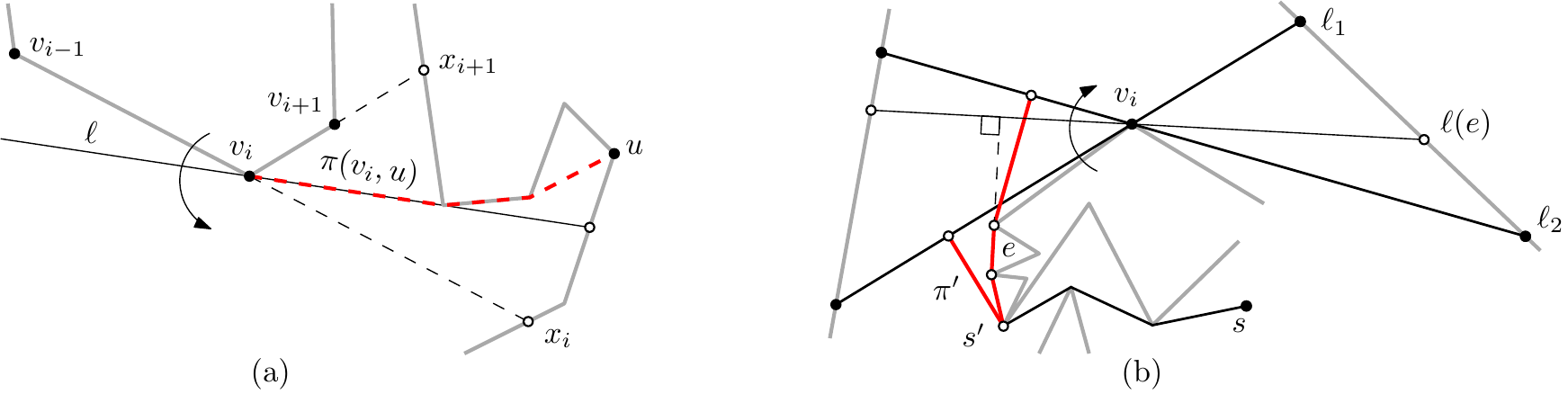}
  \caption{(a) The line-of-sight intersecting $\pi(v_i,u)$
      contains the edge of $\pi(v_i,u)$ incident to $v_i$. (b) The
      maximal common subpath of $\pi(s,\ell_1)$ and $\pi(s,\ell_2)$
      from $s$ is $\pi(s,s')$; $\pi'=\pi(s',\ell_1)\cup\pi(s',\ell_2)$
      (blue). }
  \label{fig:query-bend}
\end{figure}

Therefore, we can find two consecutive events with an optimal solution
between them, and we can obtain an optimal solution in $O(\log^2 n)$
time in total.

\begin{theorem}\label{theorem:query}
  Given a simple $n$-gon $P$, we can preprocess it in $O(n)$ time to
  find the minimum of the longer distance that $s$ and $t$ travel in
  order to see each other in $P$ can be computed in $O(\log^2 n)$ time
  for any two query points $s,t\in P$.
\end{theorem}

\section{Conclusions and Open Problems}
We have presented a linear time algorithm that solves two variants of
the quickest pair-visibility problem for a simple polygon: either we
want to minimize the maximum length of a traveled path or we want to
minimize the sum of the lengths of both traveled paths.
	
Additionally, we have considered a query version of the
quickest-visibility problem for the min-max variant. We can preprocess
a simple $n$-gon in linear time so that the minimum of the longer
distance the two query points travel can be computed in $O(\log^2 n)$
time for any two query points.

We conclude this paper with some interesting open problems.
\begin{enumerate}
\item Is there a way to extend our algorithm to more than two query
  points? More precisely, given $k$ points in a simple
    polygon, compute the minimum distance that these points must
    travel in order to see each other (at the same moment).
\item Find an efficient algorithm for the query version of the
  quickest-visibility problem for the min-sum problem.
\end{enumerate}

\section*{Acknowledgments}
This research was initiated at the 19th Korean Workshop on Computational Geometry in W\"urzburg, Germany.

\bibliographystyle{elsarticle-num}

\bibliography{romeojulietpaper} 

\clearpage

\end{document}